\documentclass[aps,superscriptaddress,twocolumn,10pt,prx]{revtex4-2}
\usepackage{silence}
\usepackage{amsmath,mathtools,amsthm,amssymb}
\usepackage{tabularx}
\usepackage{tabularray}
\usepackage{tabularx}
\usepackage{tabularray}

\usepackage{graphicx}

\usepackage{color}
\usepackage{bbold}
\usepackage{enumitem}
\usepackage{centernot}
\usepackage{complexity}
\usepackage{physics}
\usepackage{comment}
\usepackage{array}
\usepackage{graphics}
\usepackage{wrapfig}
\usepackage{fontsize}
\graphicspath{{figures/}}
\usepackage{biolinum}
\usepackage{bbm}
\usepackage{dsfont}
\usepackage{mathrsfs}
\usepackage{mathdots}
\usepackage{enumitem}
\usepackage[table]{xcolor}
\usepackage{booktabs}
\usepackage{amssymb}

\definecolor{lightred}{RGB}{243,229,231}
\definecolor{lightgreen}{RGB}{241,255,239}
\definecolor{lightblue}{RGB}{232,240,244}
\definecolor{RoyalBlue}{RGB}{65,105,225}
\definecolor{ForestGreen}{RGB}{34,139,34}   
\definecolor{Maroon}{RGB}{135,0,0}
\definecolor{myrefcolor}{rgb}{0.067,0.5,0.5}
\definecolor{myurlcolor}{rgb}{0.1,0,0.9}

\makeatletter
\newcommand\stoptoc{%
  \let\addtocontents@orig\addtocontents
  \renewcommand{\addtocontents}[2]{}%
}
\newcommand\resumetoc{\let\addtocontents\addtocontents@orig}
\makeatother
\usepackage[
    breaklinks,
    pdftex,
    colorlinks=true,
    linkcolor=myrefcolor,
    citecolor=myrefcolor,
    urlcolor=myrefcolor
]{hyperref}
\usepackage[capitalize]{cleveref}

\newcommand{\haar}[0]{\mathrm{Haar}}
\usepackage{hyperref}

\newtheorem{theorem}{Theorem}

\newtheorem{lemma}{Lemma}
\newtheorem{corollary}{Corollary}
\newtheorem{definition}{Definition}

\begin{document}

\title{Fermionic entropy: an efficiently measurable strong monotone for non-Gaussianity}

\author{Lorenzo Leone}

\thanks{These authors contributed equally.}

\affiliation{Dipartimento di Ingegneria Industriale, Università degli Studi di Salerno, Via Giovanni Paolo II, 132, 84084 Fisciano (SA), Italy}
\affiliation{INFN, Sezione di Napoli, Gruppo Collegato di Salerno, Italy}

\author{Lennart Bittel}
\thanks{These authors contributed equally.}

\affiliation{Dahlem Center for Complex Quantum Systems, Freie Universit\"at Berlin, 14195 Berlin, Germany}

\begin{abstract}
Fermionic Gaussian states form a central class of classically tractable quantum states, while fermionic non-Gaussianity provides the resource required to go beyond free-fermion dynamics. A key challenge is to quantify this resource through monotones that are both mathematically rigorous and experimentally accessible. Here, we show that the fermionic entropy, defined through the squared Frobenius norm of the correlation matrix, is a strong pure-state Gaussian monotone. Its simple closed-form expression also makes it directly measurable: we show that the associated fermionic purity can be unbiasedly estimated up to additive error $\varepsilon$ using $O(\varepsilon^{-2})$ two-copy measurements, independently of the system size. Moreover, we prove that the fermionic entropy obeys asymptotic continuity and, as a direct consequence, establish its operational meaning as the upper bound to the asymptotic rate of non-Gaussianity distillation. We further derive a linear sample complexity bound for tolerant testing of fermionic Gaussian states, providing a quadratic improvement over the state of the art. As a further application of our results, we study unitary designs generated by Matchgate circuits supplemented with Majorana-local non-Gaussian gates. We prove that a linear number of such gates is necessary even to achieve an approximate state $2$-design with error below $0.4\%$. Combined with known nearly linear upper bounds for relative-error designs, this determines the optimal doping level, up to logarithmic factors, across all relevant design notions and reveals the extensive non-Gaussianity cost required to generate Haar-like quantum dynamics in this architecture. 
\end{abstract}
\maketitle
\stoptoc

{\em Introduction.---} Quantum advantage requires quantum states and dynamics that cannot be efficiently reproduced by classical methods~\cite{eisert2026mindgapsfraughtroad}. Yet not all quantum systems are equally hard to simulate. Important families of highly entangled states remain classically tractable because of additional algebraic structure. Two central examples are stabilizer states, generated by Clifford circuits, and fermionic Gaussian states, generated by Matchgate circuits~\cite{valiant2001quantum,knill2001fermioniclinearopticsmatchgates,bravyi_fermionic_2002,Jozsa_2008,TerhalDiVincenzo2002,PhysRevA.102.052604,langer2026matchgatecircuitrepresentationfermionic}. Understanding which resources allow a quantum system to escape these tractable families is therefore a basic question in quantum information.

Fermionic Gaussian states, also known as free-fermionic states, play a particularly broad role. They arise naturally from Hamiltonians that are quadratic in fermionic operators and can be efficiently represented through their correlation matrix~\cite{bravyi_fermionic_2002,Surace_2022}. 
Interactions, however, generally take the system beyond this free description. The resulting fermionic non-Gaussianity captures correlations that cannot be explained by any free-fermionic model and provides the resource needed to promote Matchgate computation towards universality~\cite{PhysRevLett.123.080503}.

The resource theory of fermionic non-Gaussianity~\cite{Sierant_2026,tarabunga2026computablemeasuresfermionicnongaussianity,haug2026practicaltestswitnessesfermionic} makes this distinction precise by identifying fermionic Gaussian states as free states, Gaussian protocols as free operations, and states outside the convex hull of fermionic Gaussian states as resourceful. As in any resource theory, a central challenge is to identify \textit{resource monotones}: non-negative functions that (i) vanish exactly on the set of free states and (ii) do not increase under free operations.

General constructions based on distances, ranks, or robustness are available~\cite{Gottlieb_2005,gottlieb2006propertiesnonfreenessentropymeasure,cudby2025gaussiandecompositionmagicstates,Dias_2024,Reardon_Smith_2024}, but they often require difficult optimizations over the full set of free states~\cite{leone2026unbearablehardnessdecidingmagic}. This makes them of limited use for large quantum systems and, especially, experiments. A similar difficulty arises in the resource theory of magic, where efficiently computable quantities such as stabilizer entropies were introduced only recently~\cite{leone_stabilizer_2022,leoneStabilizerEntropiesAre2024}. For fermionic non-Gaussianity, recent work has introduced computable measures based on the fermionic commutant~\cite{Sierant_2026,falcão2026fermionicmagicresourcesdisordered}, based on fermionic convolution~\cite{lyu2024fermionicgaussiantestingnongaussian,coffman2025measuringnongaussianmagicfermions} as well as computable monotones derived from the correlation matrix~\cite{tarabunga2026computablemeasuresfermionicnongaussianity}. Their status as resource monotones, however, has only recently begun to be understood. 

In this work, we establish the {\em strong monotonicity} of the simplest experimentally accessible covariance-based measure of fermionic non-Gaussianity, which we refer to as the \emph{fermionic entropy}, defined by
\begin{align}
M_f = n\left(1-P_f\right),\quad P_f\coloneqq\frac{\|\Gamma(\psi)\|_2^2}{2n}
\end{align}
where $n$ is the number of qubits, $\Gamma(\psi)$ is the correlation matrix of a pure $n$-qubit state $\psi$, $\|\cdot\|_2$ denotes the Frobenius norm, and $P_f$ is referred to as the \textit{fermionic purity}. We note that the fermionic entropy coincides with the quadratic generalized one-body entropy~\cite{Gigena_2016}, and later introduced in Ref.~\cite{Sierant_2026} as a particular case of a larger family of measures, under the name of {\em fermionic antiflatness}. It was then identified with the \textit{$2$-occupation number entropy}~\cite{tarabunga2026computablemeasuresfermionicnongaussianity}.

The fermionic entropy can be readily recast as the expectation value of a Hermitian operator, which makes the monotone experimentally accessible~\cite{Sierant_2026}. By carefully bounding the moments of this Hermitian operator, we refine the state-of-the-art measurement protocol available in the literature~\cite{tarabunga2026computablemeasuresfermionicnongaussianity}, which is tailored to the direct reconstruction of the full correlation matrix and whose measurement cost grows with the system size. We show that $P_f$ can be {\em unbiasedly estimated} up to error $\varepsilon$ using only $O(\varepsilon^{-2})$ $2$-copy measurements, with no dependence on the number of qubits. As a consequence, we derive an improved algorithm for tolerant testing of fermionic Gaussian states requiring only a linear number of samples in the system size, achieving a quadratic improvement over the state-of-the-art~\cite{tarabunga2026computablemeasuresfermionicnongaussianity}.

Moreover, we prove that the fermionic entropy $M_f$ obeys {\em asymptotic continuity} and, as an immediate bookkeeping consequence~\cite{chitambar_quantum_2019}, establish its operational meaning as the upper bound to the approximate achievable asymptotic rate for non-Gaussianity resource distillation.

A further main result of our work is the use of the fermionic purity to study how fast doped Matchgate circuits converge to unitary designs, which have become a cornerstone for modeling complex quantum dynamics, including thermalization and information scrambling~\cite{popescu_entanglement_2006,sekino_fast_2008,Munson2025Mar}. A $t$-doped Matchgate circuit alternates free Matchgate layers with $t$ Majorana-local non-Gaussian gates, interpolating between efficiently simulable free-fermion dynamics and universal quantum computation~\cite{Dias_2024}. In direct analogy with the Clifford setting~\cite{haferkamp2020QuantumHomeopathyWorks,bittel2026adaptive},  we ask how much doping is required for Matchgate circuits to form quantum designs. We show that any ensemble with fewer than a linear number of Majorana-local doping gates fails to form even an approximate state $2$-design with accuracy better than $0.4\%$. Conversely, known shallow-circuit constructions imply that a nearly linear number of local non-Gaussian gates is sufficient to form relative-error designs of constant degree~\cite{SchusterHaferkampHuang2025}. Hence, up to logarithmic factors, we determine the optimal doping level to be $t=\widetilde{\Theta}(n)$. Unlike doped Clifford circuits, for which system-size-independent constructions exist~\cite{haferkamp2020QuantumHomeopathyWorks,bittel2026adaptive}, doped Matchgate circuits require an extensive amount of non-Gaussianity to form unitary designs.

To summarize, the main results of this work are:

\begin{enumerate}[label=(\roman*)]
\item The fermionic entropy $M_f$ is a strong pure-state monotone under Gaussian protocols.

\item The fermionic purity $P_f$ can be unbiasedly estimated up to additive error $\varepsilon$ using $O(\varepsilon^{-2})$ measurements, independently of the system size.

\item The fermionic entropy $M_f$ is asymptotically continuous and provides an upper bound on the achievable asymptotic rate of non-Gaussianity distillation.

\item There exists a tolerant testing for fermionic Gaussian states algorithm with sample complexity $\tilde{O}(n)$.
\item A linear number of Majorana-local doping gates is necessary and, up to logarithmic factors, sufficient to form quantum designs of constant degree.
\end{enumerate}

\medskip
\textit{Setup and notation.---} We consider a system of $n$ qubits with Hilbert space $\mathbb{C}^{2\otimes n}$ and operator basis given by the Pauli group $\mathbb{P}_n$. We denote $\gamma_1,\ldots, \gamma_{2n}$ the Majorana operators satisfying $\{\gamma_a,\gamma_b\}=2\delta_{ab}\mathbb{I}$. We denote by $\mathcal{M}_n$ the group of Matchgates, which is the group generated by product of unitaries $e^{\theta\gamma_a\gamma_b}$. Free-fermionic or Gaussian states are pure states obtained from $\ket{0}^{\otimes n}$ by the action of $U\in \mathcal{M}_n$. We further denote by $\|\cdot\|_1$ and $\|\cdot\|_{2}$ the trace and the Frobenius norm of operators and by $\|\cdot\|_\diamond$ the diamond norm for quantum channels. For an ensemble $\mathcal{E}$ on the unitary group $\mathcal{U}_n$, $U \sim \mathcal{E}$ indicates that $U$ is drawn uniformly at random from $\mathcal{E}$; in particular, $\haar$ denotes the Haar (uniform) measure.

\medskip
{\em Fermionic non-Gaussianity.---} Any resource theory starts from the dichotomy between \textit{free} and \textit{resourceful} states~\cite{chitambar_quantum_2019}. In the resource theory of fermionic non-Gaussianity, free states are defined as the convex hull of \textit{free-fermionic Gaussian states}, which we denote by $\mathrm{GHull}$~\cite{Melo_2013,Vershynina_2014,ramkumar2026hightemperaturefermionicgibbsstates}. This choice is motivated~\cite{tarabunga2026computablemeasuresfermionicnongaussianity} by the corresponding choice of \textit{free operations}---the other basic ingredient of a resource theory---namely, \textit{Gaussian protocols}. These are completely positive and trace-preserving maps built from the following elementary free operations: (i) Matchgate unitaries; (ii) partial trace; (iii) measurements in the computational basis; (iv) composition with free-fermionic states; and (v) the previous operations conditioned on measurement outcomes. Since conditioned operations may produce different output states, the most general form of a Gaussian protocol acting on a state $\rho$ is $\mathcal{E}(\rho)={(p_i,\rho_i)}$, namely, a collection of $n_i$-qubit states $\rho_i$ occurring with probabilities $p_i$. A protocol $\mathcal{E}$ is \emph{deterministic} if $\mathcal{E}(\rho)=\tilde{\rho}$, meaning that a unique quantum state is obtained with unit probability.

Given the class of Gaussian protocols, one can define monotones for the resource theory of fermionic non-Gaussianity.
\begin{definition}[Gaussian monotones]\label{def:gaussianmonotone}
\leavevmode\par
\begin{itemize}
    \item A Gaussian monotone $\mathcal{M}$ is a real-valued function defined for all $n\in \mathbb{N}$ qubit systems, or collections thereof, such that: (i) $\mathcal{\mathcal{M}}(\rho)=0$ if and only $\rho\in\mathrm{GHull}$; and (ii) $\mathcal{M}$ is nonincreasing under Gaussian protocols $\mathcal{G}$, i.e. $\mathcal{M}(\mathcal{E}(\rho))\le \mathcal{M}(\rho)$.
\item A pure-state Gaussian monotone instead satisfies $\mathcal{M}(\psi)\le \mathcal{M}(\phi)$ for any pair of pure states $\psi,\phi$ such that there exists $\mathcal{E}\in\mathcal{G}$ with $\mathcal{E}(\phi)=\psi$.

\item A pure-state Gaussian monotone is said to be {strong} if, for every pure state $\ket{\psi}$ and every Gaussian protocol $\mathcal{E}\in\mathcal{G}$ producing the ensemble of pure states $\{(p_i,\ket{\phi_i})\}$, it holds that $\mathcal{M}(\psi)\ge \sum_i p_i\mathcal{M}(\phi_i)$.
\end{itemize}
\end{definition}

In other words, a pure-state monotone is monotonic under \textit{deterministic} pure-state Gaussian protocols, namely, protocols that map pure states to pure states. Since Gaussian protocols may involve several elementary operations, intermediate states may become mixed, as long as the final state is pure. Whereas, strong monotonicity states that the average non-Gaussianity cannot increase under a nondeterministic Gaussian protocol. 

\medskip
{\em A measurable strong Gaussian monotone.---}
To any $n$-qubit state $\rho$, one can associate a correlation matrix $\Gamma(\rho)$, defined as the $2n\times 2n$ antisymmetric matrix
\begin{align}\label{eq:cormatrix}
\Gamma_{ij}(\rho)\coloneqq -\frac{i}{2}\tr([\gamma_i,\gamma_j]\rho)\,.
\end{align}
Free-fermionic states are uniquely characterized by their correlation matrix: the expectation value of every product of Majorana operators, or equivalently of every Pauli operator, is determined by $\Gamma(\psi)$ through Wick's theorem~\cite{bravyi_fermionic_2002,Surace_2022}. Motivated by the identity $\|\Gamma(\sigma)\|_2^2=2n$ for every pure free-fermionic state $\sigma$, one can use the squared Frobenius norm of the correlation matrix to quantify how much of a pure state is captured by its degree-$2$ Majorana correlators.

\begin{definition}[Fermionic entropy]\label{def:fermionicentropy}
The fermionic entropy of a pure quantum state $\ket{\psi}$ is defined as
\begin{align}
M_f(\psi)\coloneqq n(1-P_f(\psi))
\end{align}
where $P_f(\psi)\coloneqq\|\Gamma(\psi)\|_2^2/(2n)$ is the \textit{fermionic purity}.
\end{definition}
The fermionic entropy satisfies the following properties~\cite{Sierant_2026,tarabunga2026computablemeasuresfermionicnongaussianity}: (i) $M_f(\psi)=0$ if and only if $\psi$ is a pure fermionic Gaussian state; (ii) $M_f(U\psi U^{\dagger})=M_f(\psi)$ for every $U\in\mathcal{M}_n$; and (iii) it is additive, namely, $M_f(\psi\otimes \phi)=M_f(\psi)+M_{f}(\phi)$ for $\psi,\phi$ having the same fermionic parity. 

Although $M_f$ was previously analyzed~\cite{Sierant_2026,tarabunga2026computablemeasuresfermionicnongaussianity}, it remained unclear whether it defines a Gaussian monotone and therefore provide a useful tool for the resource theory of fermionic non-gaussianity. Our first main result answers this question in the affirmative. The following theorem shows that $M_f$ is, in fact, a \textit{strong} pure-state Gaussian monotone according to \cref{def:gaussianmonotone}.

\begin{theorem}\label{th:monotonicity}
The fermionic entropy $M_f$ is a strong pure-state Gaussian monotone. Moreover, its convex-roof extension
\begin{align}
    \widetilde{M}_f(\rho)=\inf_{\rho=\sum_jp_j\phi_j}\sum_jp_jM_f(\phi_j)
\end{align}
is a strong monotone for arbitrary mixed states $\rho$.
\end{theorem}
\noindent
{\em Proof sketch.} We closely follow the strategy used by the same authors in Ref.~\cite{leoneStabilizerEntropiesAre2024}. For a pure state $\ket{\psi}$, it is enough to consider decompositions of the form $\ket{\psi}=\sqrt{p}\ket{0}\ket{\phi_0}+\sqrt{1-p}\ket{1}\ket{\phi_1}$ and prove that $M_f(\psi)\ge pM_f(\phi_0)+(1-p)M_f(\phi_1)$. The claim then follows from Theorem 3 of Ref.~\cite{leoneStabilizerEntropiesAre2024}. See \cref{app:proofth1} for the complete proof.\qed
\medskip

A key feature of the fermionic entropy is that it admits a closed-form expression in terms of the correlation matrix (\cref{eq:cormatrix}) and can be expressed as the expectation value of a Hermitian operator: defining $\Lambda\coloneqq\sum_{i}\gamma_{i}^{\otimes 2}$~\cite{bravyi_fermionic_2002}, one finds
\begin{align}\label{eq:measurement}
\tr(\Lambda^2\rho^{\otimes 2})=2M_{f}(\rho)\,,
\end{align}
which makes the fermionic entropy readily experimentally measurable.
Moreover, Refs.~\cite{Sierant_2026,tarabunga2026computablemeasuresfermionicnongaussianity} observed that \cref{eq:measurement} can be estimated via Bell sampling, since $\Lambda$ is diagonal in the Bell basis of two copies of the Hilbert space. Unfortunately, the operator norm satisfies $\|\Lambda^2\|_{\infty}=4n^2$, so a naive analysis of the sample complexity for estimating \cref{eq:measurement} would suggest a dependence on the system size. However, one of our key technical lemmas (\cref{lem:variancebounded} in \cref{boundedvarianceapp}) underlying the proof of all the results, shows that
\begin{align}\label{eq:boundsmoment}
\tr(\Lambda^{2r}\rho^{\otimes 2})=O(n^{r})
\end{align}
reflecting the fact that the eigenvectors corresponding to the largest eigenvalues of $\Lambda$ are necessarily entangled. This immediately yields the following main result.

\begin{theorem}\label{th:measurementpurity}
There exists an unbiased estimator that, acting on two copies of $\rho$, estimates the fermionic purity up to additive error $\varepsilon$ and failure probability $\delta$ using $O(\varepsilon^{-2}\log\delta^{-1})$ two-copy measurements. Consequently, $M_f(\rho)$ can be estimated up to additive error $\varepsilon$ using $O(n^2\varepsilon^{-2}\log\delta^{-1})$ measurements.
\end{theorem}

\noindent
\emph{Proof.} The proof follows immediately from two observations. First, $\frac{1}{2n}\tr(\Lambda^2\rho^{\otimes 2})$ can be estimated via Bell sampling on $\rho^{\otimes 2}$~\cite{Sierant_2026,tarabunga2026computablemeasuresfermionicnongaussianity}. Second, a corollary of \cref{lem:variancebounded}, \cref{cor:tailbound}, shows that the moment-generating function $\tr(e^{t\Lambda^2/n})\le (1-16t)^{-1}$ is bounded, which implies the claimed sample complexity by a standard Bernstein--Chernoff bound for the empirical mean of independent samples.\qed

The next theorem establishes the fermionic entropy as an asymptotically continuous Gaussian monotone~\cite{Synak_Radtke_2006}, by proving a Fannes-like inequality.

\begin{theorem}\label{cor:robustness}
Let $\rho$ and $\rho'$ be two arbitrary quantum states. Then, the fermionic entropy obeys the following continuity bound:
\begin{align}\label{eq:robustness}
|M_f(\rho)-M_f(\rho')|\le 2n\|\rho-\rho'\|_1,.
\end{align}
\end{theorem}
See \cref{cor:robustnessapp} for the proof. We note that \cref{eq:robustness} drastically improves upon the  continuity bound  on $M_f(\rho)$ of Ref.~\cite{tarabunga2026computablemeasuresfermionicnongaussianity}, which scales quadratically with the system size. 

As is customary in resource theories, addivity and strong monotonicity together with asymptotic continuity implies that the monotone under consideration upper bounds the ultimate rate of resource conversion under free operations. Therefore, our main results, \cref{th:monotonicity,cor:robustness}, immediately yields the following corollary, whose proof follow by standard arguments~\cite{chitambar_quantum_2019}.
\begin{corollary}
 Let $\psi$ and $\phi$ be pure states with $M_{f}(\phi)>0$ and $N\in\mathbb{N}$. Any sequence of Gaussian protocols converting $\psi^{\otimes N}$ to $\phi^{\otimes R_N\times N}$ with vanishing trace-distance error $\varepsilon_N$ (i.e. $\lim_{N}\varepsilon_N=0$) satisfies
 \begin{align}
\lim\sup R_N\le \frac{M_{f}(\psi)}{M_{f}(\phi)}
 \end{align}
 More generally, for probabilistic protocols, the expected asymptotic yield is bounded
by the same ratio.
\end{corollary}

As a consequence of \cref{th:measurementpurity,cor:robustness,eq:boundsmoment}, we improve the state-of-the-art sample complexity for \emph{tolerant} property testing of pure fermionic Gaussian states by a quadratic improvement in the system size~\cite{tarabunga2026computablemeasuresfermionicnongaussianity}. Given an unknown state $\rho$, the goal is to decide whether it is within trace distance $\varepsilon_A$ of some pure fermionic Gaussian state, or whether it is at least $\varepsilon_B$ away in trace distance from every pure fermionic Gaussian state.

\begin{theorem}\label{cor:propertytesting}
For any $\varepsilon_A\le \frac{\varepsilon_B^2}{64n}$, there exists an algorithm that, through the efficient measurement of the fermionic purity, uses $O(n\varepsilon_B^{-2}\log [n\varepsilon_B^{-2}]\log\delta^{-1})$ copies of $\rho$, succeeds with probability at least $1-\delta$, and solves the tolerant fermionic Gaussian testing problem.
\end{theorem}
See \cref{cor:propertytestingapp} for a proof. In the next sections, we use our main technical results discussed above to derive tight bounds on the formation of quantum designs from Matchgates as free resources, in close analogy with $t$-doped Clifford circuits, and solving an open problem in the literature~\cite{trigueros2026unitarydesignsdopedmatchgate}.

\medskip
{\em $t$-doped Matchgates.---} Starting from the state $\ket{0}$, Matchgate circuits can generate only {free-fermionic} states~\cite{Jozsa_2008,Terhal_2002,PhysRevA.102.052604}. Since the Matchgate group is not universal, universality can be achieved by ``doping'' it with Majorana rotations of the form $W=e^{\theta\prod_{j=1}^{\kappa}\gamma_{i_j}}$ with $\kappa=3,4$. However, generating an arbitrary unitary to constant precision requires an exponential number of such non-Gaussian gates. This motivates the intermediate class of $t$-doped Matchgate circuits, defined as
\begin{align}
U_t=G_tW_tG_{t-1}W_{t-1}\cdots W_1G_0\,,
\end{align}
where the Matchgate layers $G_i\in\mathcal{M}_n$ are interleaved with non-Gaussian $\kappa$-local gates $W_i$ with $\kappa=3,4$~\cite{Dias_2024}. The corresponding states $\ket{\psi_t}=U_t\ket{0}$ are called $t$-doped free-fermionic, or Gaussian, states~\cite{MeleHerasymenko2025}. In direct analogy with Clifford circuits and doped stabilizer states, it is natural to ask how many non-Gaussian gates are needed before Matchgate circuits begin to reproduce features of universal quantum dynamics~\cite{paviglianiti2025emergencegenericentanglementstructure}. In this work, we address this question through the moments of the Haar distribution, i.e. unitary designs, briefly introduced in the next section.

\medskip
{\em Unitary designs.---} We  review the notions of unitary designs used below; see Ref.~\cite{bittel2026adaptive} for a complete overview. Let $\mathcal{E}$ be an ensemble of unitaries on $n$ qubits, with $k$-moment operator $\Phi_{\mathcal{E}}(\cdot)\coloneqq\mathbb{E}_{U\sim\mathcal{E}}U^{\otimes k}(\cdot)U^{\dagger\otimes k}$. The ensemble $\mathcal{E}$ is an {\em additive-error $\varepsilon$-approximate unitary $k$-design} if~\cite{Harrow2009Random2-designs,brandao_local_2016}
\begin{align}\label{eq:additivedesign}
\|\Phi_{\mathcal{E}}-\Phi_{\haar}\|_{\diamond}\leq \varepsilon,.
\end{align}
The corresponding notion of {\em $\varepsilon$-approximate state $k$-design} is obtained by acting on the all-$\ket{0}$-state:
\begin{align}\label{statedesigns}
\|\Phi_{\mathcal{E}}(\ketbra{0}{0}^{\otimes k})-\Phi_{\haar}(\ketbra{0}{0}^{\otimes k})\|_1\leq \varepsilon,.
\end{align}

A stronger notion is that of a {\em relative-error $\varepsilon$-approximate unitary $k$-design}~\cite{brandao_local_2016}, defined by
\begin{align}\label{eq:relative designs}
(1-\varepsilon)\Phi_{\haar}\leq \Phi_{\mathcal{E}}\leq (1+\varepsilon)\Phi_{\haar},,
\end{align}
where $\leq$ is the ordering of completely positive maps. Relative-error designs are secure against arbitrary measurements, but are stronger than operationally necessary, since relative differences between channels need not always be detectable.

The most operationally meaningful, yet also the most difficult to analyze, is the notion of a \textit{quantum-secure design}~\cite{bittel2026adaptive}. Let $\boldsymbol{V}=(V_1,\ldots,V_k)$ be arbitrary unitaries acting on the system and an ancilla of $n'$ qubits, and define $\ket{\Psi_U(\boldsymbol{V})}\coloneqq UV_kUV_{k-1}\cdots UV_1\ket{0}$. The ensemble $\mathcal{E}$ is an {\em $\varepsilon$-approximate quantum-secure unitary $k$-design} if
\begin{align}\label{quantumsecuredesign}
\sup_{n'}\max_{\boldsymbol{V}}\left\|\underset{U\sim\mathcal{E}}{\mathbb{E}}\Psi_U(\boldsymbol{V})-\underset{U\sim\haar}{\mathbb{E}}\Psi_U(\boldsymbol{V})\right\|_1\leq \varepsilon,.
\end{align}
This quantity is the maximum distinguishing advantage of any quantum experiment making at most $k$ queries to $U$. The notions satisfy
\begin{align}\label{hierarchydesigns}
\text{relative}-&\text{error}\\
\nonumber\Downarrow\\
\nonumber\text{quantum}-&\text{secure}\\
\nonumber\Downarrow\\
\nonumber\text{additive}-&\text{error}\\
\nonumber\Downarrow\\
\nonumber\text{state}-&\text{design}
\end{align}
and all converse implications are false.

The distinction is important for doped circuits: for Clifford circuits, relative-error unitary $k$-designs require a doping level $\Theta(nk)$~\cite{leone2025noncliffordcostrandomunitaries}, whereas quantum-secure unitary $k$-designs can be obtained with only $O(k^2)$ non-Clifford gates, independently of $n$~\cite{bittel2026adaptive}.

\medskip
{\em Unitary designs with doped Matchgate circuits.---} We now present one of the main consequences of our results: a linear lower bound on the number of doping gates required to form unitary designs.

Ref.~\cite{trigueros2026unitarydesignsdopedmatchgate}, building on Ref.~\cite{sierant2026theorymatchgatecommutant}, showed that doped Matchgate circuits form approximate state $2$-design with $t=O(n)$ gates and relative unitary $2$-design with $t=O(n\operatorname{poly\log}n)$ gates. A complementary lower bound was proved in Ref.~\cite{Tarabunga2026}: $t=\Omega\left(\sqrt{n\log(1/\varepsilon)}\right)$ gates are necessary to form an approximate state $2$-design with error $\varepsilon$. Hence, whether a sublinear doping level could suffice therefore remained open.

As the main application of our results, we close this gap and prove that the cost is extensive. Let $\mathcal{E}_{t}$ be an ensemble of $t$-doped Matchgate circuits. We ask which values of $t$ are necessary and sufficient for $\mathcal{E}_{t}$ to form an $\varepsilon$-approximate unitary design according to any of the notions introduced above. We answer this question through matching upper and lower bounds, up to logarithmic factors. For the upper bound, we consider the strongest notion in the hierarchy of \cref{hierarchydesigns}, while for the lower bound we consider the weakest one. Together, these results determine the required doping level for all intermediate notions.

We first discuss the upper bound. Ref.~\cite{SchusterHaferkampHuang2025} showed that, for constant design degree $k$, an $\varepsilon$-approximate relative-error unitary design can be generated by circuits of depth $O(\log \frac{n}{\varepsilon})$. Since such circuits contain at most $O(n\log \frac{n}{\varepsilon})$ gates, it follows directly that $O(n\log\frac{n}{\varepsilon})$ doping gates with $\kappa=3,4$ are sufficient to form an $\varepsilon$-approximate relative-error unitary design. It is indeed sufficient to generate arbitrary gates on the first qubit using Matchgates and Majorana rotations with locality $\kappa=3,4$, together with swap operators, each of which requires only a single $\kappa=4$ Majorana rotation. By descending the hierarchy in \cref{hierarchydesigns}, the same ensemble also forms quantum-secure, additive-error, and state $k$-designs with doping level $\widetilde{O}(n)$.

We now show that the above upper bound is optimal up to logarithmic factors. The key ingredient is the efficient measurement of the fermionic purity in \cref{th:measurementpurity}. 

\begin{theorem}\label{cor:designs}
Let $n\ge 7$. For any $t<n/(8\kappa)$, no ensemble $\mathcal{E}_t$ of $t$-doped Matchgate circuits can form an $\varepsilon$-approximate state $2$-design for any $\varepsilon\le 2^{-8}$. By ascending the hierarchy above, it therefore cannot form an additive-error, quantum-secure, or relative-error unitary $k$-design for any $k\ge 2$.
\end{theorem}
\noindent
{\em Proof sketch.} It is sufficient to observe that $\mathbb{E}_{\haar}[P_f]=\exp(-\Omega(n))$, whereas every state $\ket{\psi_t}=U_t\ket{0}^{\otimes n}$ with $U_t\in\mathcal{E}_t$ satisfies $P_f(\psi_t)\ge 1-O(t/n)$. Hence, the ensemble remains distinguishable from the Haar ensemble at the level of second moments and, consequently, cannot form an $\varepsilon$-approximate state $2$-design. See \cref{cor:designapp} for the complete proof.\qed
\medskip

This result shows that reproducing Haar-like dynamics is costly for Matchgate circuits supplemented by Majorana-local gates. It also reveals a sharp difference between two central classes of classically simulable circuits: Clifford circuits and Matchgate circuits.

\medskip
{\em Discussion and conclusions.---} In this work, we extensively studied the fermionic entropy, arguably the simplest measure of fermionic non-Gaussianity, and proved several key properties of a good resource monotone: (i) strong monotonicity, (ii) efficient measurability, and (iii) asymptotic continuity. This then allowed us to improve the state-of-the-art sample complexity for tolerant fermionic Gaussianity testing, and to determine the optimal scaling of non-Gaussian resources that, together with free Matchgate circuits, generate unitary $k$-designs.

However, this work leaves several important open questions. First, although we improve the sample complexity for tolerant fermionic Gaussianity testing from $O(n^2)$ to $O(n\log n)$, it remains an open question whether this can be further improved, and in particular whether an $O(1)$-sample tester exists. Second, while the strong monotonicity of $M_f$ allows one to extend it to a strong mixed-state non-Gaussianity monotone, it is natural to ask whether a more direct or intrinsically defined mixed-state monotone exists. While scalable witnesses have recently been constructed~\cite{tang2026witnessexpansionunifiedframework}, they do not provide an intrinsic mixed-state monotone. Third, although the fermionic entropy belongs to the broader family of fermionic non-Gaussianity measures (fermionic antiflatness) introduced in Ref.~\cite{Sierant_2026}, it remains unclear whether our proof strategy can be extended to prove or disprove the monotonicity of the entire family. We hope that the present work will stimulate further progress on these questions and deepen our understanding of fermionic non-Gaussianity as a quantum resource.

{\em Acknowledgments.---} The authors thank Xhek Turkeshi, Piotr Sierant, Antonio A. Mele, Salvatore F.E. Oliviero, Yaroslav Herasymenko, Poetri S. Tarabunga and Jens Eisert for feedbacks and discussions. LB is funded by the
European Research Council (DebuQC).

{\em Note added.---} Shortly before posting this manuscript, we became aware that the recently posted v2 of Ref.~\cite{tarabunga2026computablemeasuresfermionicnongaussianity} independently contains a proof of \cref{th:monotonicity}. Moreover, we learned through private communication that a forthcoming version of Ref.~\cite{haug2026practicaltestswitnessesfermionic} will present a measurement scheme for $M_f$ with sample complexity $O(n^2\varepsilon^{-2})$, thereby matching the scaling established in \cref{th:measurementpurity}.

\let\oldaddcontentsline\addcontentsline
\renewcommand{\addcontentsline}[3]{}

\let\addcontentsline\oldaddcontentsline
\appendix
\resumetoc
\onecolumngrid
\clearpage
\begin{center}
    {\normalfont\Large\bfseries Appendix}
\end{center}
\setcounter{secnumdepth}{2}
\setcounter{equation}{0}
\setcounter{figure}{0}
\setcounter{table}{0}
\setcounter{section}{0}
\renewcommand{\thetable}{S\arabic{table}}

\renewcommand{\thefigure}{S\arabic{figure}}
\renewcommand{\thesection}{S.\arabic{section}} 
\counterwithout{equation}{section}
\renewcommand{\theequation}{S\arabic{equation}}
\tableofcontents

\section{Strong monotonicity of the fermionic entropy: proof of \cref{th:monotonicity}}\label{app:proofth1}
In this section, we prove \cref{th:monotonicity}. To prove it, we closely follow the strategy of Ref.~\cite{leoneStabilizerEntropiesAre2024}. We therefore proceed and prove the following fundamental lemma.
\begin{lemma}\label{lemma:decomposition}
    Consider $\ket{\psi}=\sqrt{p}\ket{0}\ket{\phi_0}+\sqrt{1-p}\ket{1}\ket{\phi_1}$. Then, it holds that
    \begin{align}
        M_{f}(\psi)\ge pM_{f}(\phi_0)+(1-p)M_{f}(\phi_1)
    \end{align}
\end{lemma}
\begin{proof}

First, we have $\gamma_1=X_1$ $\gamma_2=Y_1$ and $\gamma_{a+2}=Z\otimes \eta_a$ for $\eta_a$ being the Majoranas for $a\in[2m]$ with $m=n-1$. Let us denote $q\coloneqq1-p$ for brevity. Let us express the covariance matrix of $\psi$ in terms of the covariance matrix of $\phi_0,\phi_1$. We have
\begin{align}
    \Gamma_{12}(\psi)=-i\langle\psi|\gamma_1\gamma_2|\psi\rangle=\langle\psi|Z_1|\psi\rangle=p-q
\end{align}
Then for $a\in[2(n-1)]$
\begin{align}
    \Gamma_{1,a+2}=-i\langle\psi|XZ\otimes \eta_a|\psi\rangle=-\langle\psi|Y\otimes \eta_a|\psi\rangle=-\sqrt{pq}i(\langle\phi_1|\eta_a|\phi_0\rangle-\langle\phi_1|\eta_a|\phi_0\rangle)=-2\sqrt{pq}\mathrm{Im}(z_a)
\end{align}
where we denoted $z_a\coloneqq\langle\phi_0|\eta_a|\phi_1\rangle$. Analogously, we find $\Gamma_{2,a+2}=2\sqrt{pq}\mathrm{Re}(z_a)$. Lastly, for $a,b\in[2(n-1)]$, we have 
\begin{align}
    \Gamma_{a+2,b+2}=p\Gamma(\phi_0)+q\Gamma(\phi_1)
\end{align}

We now compute the Frobenius norm:
\begin{align}\label{proof:fnorm}
    \frac{1}{2}\|\Gamma(\psi)\|_2^2&=\Gamma_{12}^2+\sum_{a=1}^{2m}\Gamma_{1,a+2}^2+\sum_{a=1}^{2m}\Gamma_{2,a+2}^2+\sum_{a,b}^{2m}\Gamma_{a+2,b+2}^2\\
    &=(p-q)^2+4pq\|z\|_2^2+\frac{1}{2}\|p\Gamma(\phi_0)+q\Gamma(\phi_1)\|_{2}^{2}
\end{align}
Let us denote $\Gamma(\phi_{0(1)})=\Gamma_{0(1)}$ for simplicity and let us analyse the last term in \cref{proof:fnorm}:
\begin{align}
    \|p\Gamma_0+q\Gamma_1\|_{2}^{2}=p^2\|\Gamma_0\|_2^2+q^2\|\Gamma_1\|_2^2+2pq\tr(\Gamma_0^{\dagger}\Gamma_1)
\end{align}
Then
\begin{align}
    p\|\Gamma_0\|_2^2+q\|\Gamma_1\|_2^2-\|q\Gamma_0+q\Gamma_1\|_2^2&=(p-p^2)\|\Gamma_0\|_2^2+(q-q^2)\|\Gamma_1\|_2^2-2pq\tr(\Gamma_0^{\dagger}\Gamma_1)\\&=pq(\|\Gamma_0\|_2^2+\|\Gamma_1\|_2^2-2\tr(\Gamma_0^{\dagger}\Gamma_1))\\&=pq\|\Gamma_0-\Gamma_1\|_2^2
\end{align}
where we used $(p-p^2)=(1-p)p=qp$. We can therefore express \cref{proof:fnorm} as:
\begin{align}
    \frac{1}{2}\|\Gamma(\psi)\|_{2}^{2}=(p-q)^2+4pq\|z\|_2^2+\frac{p}{2}\|\Gamma_0\|_2^2+\frac{q}{2}\|\Gamma_1\|_2^2-\frac{pq}{2}\|\Gamma_0-\Gamma_1\|_2^2
\end{align}
To show the claim, i.e. 
\begin{align}n-\frac{1}{2}\|\Gamma(\psi)\|_2^2\ge p\left[(n-1)-\frac{1}{2}\|\Gamma_0\|_2^2\right]+q\left[(n-1)-\frac{1}{2}\|\Gamma_1\|_2^2\right]\,,
\end{align}
we therefore need to show that 
\begin{align}
    1-(p-q)^2-4pq\|z\|_2^2+\frac{pq}{2}\|\Gamma_0-\Gamma_1\|_2^2\ge 0
\end{align}
Moreover, we notice that $1-(p-q)^2=1-p^2-q^2+2pq=p(1-p)+q(1-q)+2pq=4pq$. Hence, we need to show that 
\begin{align}
    1-\sum_{a=1}^{2m}|\langle\phi_0|\eta_a|\phi_1\rangle|^2+\frac{1}{8}\|\Gamma_0-\Gamma_1\|_2^2\ge 0
\end{align}
To show the claim, we proceed as follows. Let $t_a=\langle\phi_0|\eta_a|\phi_1\rangle$ be a complex vector. Expand $t=x+iy$ with $x,y\in \mathbb{R}^{2m}$. Since $\operatorname{dim}\operatorname{span}\{x,y\}\le 2$, then there exists a orthogonal matrix $O\in SO(2m)$ that maps this subspace to $\operatorname{span}\{e_1,e_2\}$ for $e_1,e_2$ being the canonical basis. Applying $O$ on the complex vector $t$ results in:
\begin{align}
    t'_{a}=(Ot)_a=\sum_{b}O_{ab}t_{b}=\sum_{b}O_{ab}\langle\phi_0|\eta_b|\phi_1\rangle=\langle\phi_0|\eta'_a|\phi_1\rangle
\end{align}
where we defined $\eta'_a=\sum_{a}O_{ab}\eta_b$. Now, by noticing that $\sum_{a=1}^{2m}|\langle\phi_0|\eta_a|\phi_1\rangle|^2=\sum_{a}|t_a|^2=\|t\|_2^2$, then $\|t\|_2^2=\|t'\|_{2}^{2}$. Moreover, let $U_O\in\mathcal{M}_n$ be the Matchgate corresponding to $O$ and let $\ket{\phi_0'}=U_O\ket{\phi_0}$ then we have
\begin{align}
    \|\Gamma_0-\Gamma_1\|_2^2= \|O^T\Gamma_0O-O^T\Gamma_1O\|_2^2= \|\Gamma_0'-\Gamma_1'\|_2^2
\end{align}
Hence, it is sufficient for us to show $\|t'\|_{2}^{2}\le 1+\frac{1}{8}\|\Gamma_0'-\Gamma_1'\|_2^2$ in the primed Majorana basis. The crucial difference in this new Majorana basis is that $t'_a=0$ for any $a\ge 3$. We therefore express
\begin{align}
    \ket{\phi_0'}&=\ket{0}\ket{v_0}+\ket{1}\ket{v_1}\\
    \ket{\phi_1'}&=\ket{0}\ket{w_0}+\ket{1}\ket{w_1}
\end{align}
We have $t'_1=\langle\phi_0'|X_1|\phi_1'\rangle=\langle  v_0|w_1\rangle+\langle v_1|w_0\rangle$, while $t'_2=\langle\phi_0'|Y_1|\phi_1'\rangle=-i\langle  v_0|w_1\rangle+i\langle v_1|w_0\rangle$. Hence
\begin{align}
    \|t'\|_2^2&=|\langle  v_0|w_1\rangle+\langle v_1|w_0\rangle|^2+|\langle  v_0|w_1\rangle-\langle v_1|w_0\rangle|^2\\&=2|\langle  v_0|w_1\rangle|^2+2|\langle v_1|w_0\rangle|^2\\
    &\le2\|v_0\|_2^2\|w_1\|_2^2+2\|v_1\|_2^2\|w_0\|_2^2\\&=2(1-\|v_1\|_2^2)\|w_1\|_2^2+2(1-\|w_1\|_2^2)\|v_1\|_2^2
\end{align}
where the last equality comes from normalization $\|\phi_0\|_2^2=\|\phi_1\|_2^2=1$. Define $x=\|v_1\|_2^2$ and $y=\|w_1\|_2^2$, so that
\begin{align}
    \|t'\|_2^2\le 2[(1-x)y+(1-y)x]
\end{align}
Now, notice that
\begin{align}
1+(x-y)^2-(1-x-y)^2&=1+x^2+y^2-2xy-(1+x^2+y^2+2xy-2x-2y)\\&=2(x+y)-4xy\\&=2[(1-x)y+(1-y)x]
\end{align}
Hence, trivially
\begin{align}
    2[(1-x)y+(1-y)x]\le 1+(x-y)^2
\end{align}
Let us now compute $\|\Gamma_0'-\Gamma_1'\|_2^2$. We have
\begin{align}
    \Gamma'_{12}(\phi_0)&=\langle\phi_0'|Z_1|\phi_0'\rangle=\|v_0\|_2^2-\|v_1\|_2^2=1-2x\\
    \Gamma'_{12}(\phi_1)&=\langle\phi_1'|Z_1|\phi_1'\rangle=\|w_0\|_2^2-\|w_1\|_2^2=1-2y
\end{align}
Hence:
\begin{align}
    \frac{1}{8}\|\Gamma'(\phi_0)-\Gamma'(\phi_1)\|_2^2\ge \frac{1}{8}(|\Gamma_{12}(\phi_0')-\Gamma_{12}(\phi_1')|+|\Gamma_{21}(\phi_0')-\Gamma_{21}(\phi_1')|)=\frac{1}{4}[2(x-y)]^2=(x-y)^2
\end{align}
which proves the desired inequality and concludes the proof. 
\end{proof}

The proof of \cref{th:monotonicity} follows directly from \cref{lemma:decomposition}. The lemma reduces any pure-state Gaussian protocol to the same elementary decomposition considered in the proof of Theorem 3 of Ref.~\cite{leoneStabilizerEntropiesAre2024}. Once this reduction is established, the remaining argument is identical: one applies the monotonicity inequality to each branch and averages over the corresponding probabilities. We therefore obtain \cref{th:monotonicity} by following the proof of Theorem 3 of Ref.~\cite{leoneStabilizerEntropiesAre2024} step by step, with \cref{lemma:decomposition} providing the only model-specific ingredient.

\section{System-size independent measurement scheme for the fermionic purity: proof of \cref{th:measurementpurity}}\label{boundedvarianceapp}
The proof of \cref{th:measurementpurity} follows directly from \cref{eq:boundsmoment}, which constitutes the key technical ingredient, together with a Bernstein--Chernoff bound for the empirical mean of independent samples, as described in the main text. We now prove the corresponding lemma.
\begin{lemma}\label{lem:variancebounded}
Let $\Lambda\coloneqq\sum_{i}\gamma_{i}\otimes \gamma_i$. Then, for every density matrix $\rho$ and for every integer $r\geq 0$, it holds that
\begin{align}
\operatorname{tr}\!\left(\Lambda^{2r}\rho^{\otimes 2}\right)
&\leq
\frac{(2r)!}{r!}\bigl((1+e)n\bigr)^r. \label{eq:main-bound-r}
\end{align}
\end{lemma}
\begin{proof}
To prove it, let us define $\Gamma_i := \gamma_i\otimes\gamma_i $. So that $[\Gamma_i,\Gamma_j]=0$ and $\Lambda=\sum_{i}^{2n}\Gamma_i$. Since $\Lambda$ is a sum of $2n$ commuting Hermitian operators, the trivial operator norm bound is
\begin{align}
\|\Lambda\| &\leq 2n,\\
\|\Lambda^{2r}\| &\leq (2n)^{2r}.
\end{align}
This is the naive $O(n^{2r})$ bound. The point of the theorem is that on product states $\rho\otimes\rho$, the expectation value is only of order $n^r$.

For a subset $S=\{i_1<\cdots<i_s\}\subseteq [2n]$ with cardinality $|S|=s$, define the ordered Majorana product
\begin{align}
\gamma_S := \gamma_{i_1}\cdots\gamma_{i_s}.
\end{align}
For the empty set, $\gamma_\varnothing=I$. Let us define its Hermitian version, which is nothing but a Pauli string:
\begin{align}
P_S := i^{s(s-1)/2}\gamma_S.
\end{align}
We define its expectation value on  a state $\rho$ as $a_S := \operatorname{tr}(P_S\rho)$, and  the sum of the expectation values of all the Pauli strings with Majorana weight $s$ as
\begin{align}
A_s \coloneqq \sum_{\substack{S\subseteq[2n]\\ |S|=s}} a_S^2. \label{eq:Asdef}
\end{align}
Equivalently, we define $\Gamma_S=\gamma_S\otimes\gamma_S$, which is equal to
\begin{align}\label{gammaspsrelation}
    \Gamma_S=(-1)^{s(s-1)/2}P_{S}\otimes P_S
\end{align}
Hence $\tr(\Gamma_S\rho^{\otimes 2})=(-1)^{s(s-1)/2}a_S^2$.

We are now ready to expand $\Lambda^{2r}
=\left(\sum_{i=1}^{2n}\Gamma_i\right)^{2r}$. A term in this expansion is specified by a word $w=(i_1,\dots,i_{2r})$, where each $i_\ell\in[2n]$, corresponding to the operator $\Gamma_{i_1}\cdots\Gamma_{i_{2r}}$. For each $j\in[2n]$, let
\begin{align}
m_j(w) := \#\{\ell\in\{1,\dots,2r\}: i_\ell=j\}. \label{eq:multiplicity-def}
\end{align}
i.e. the number of occurrences of the symbol $j$ in the word $w$. Because $\Gamma_j^2=I$, only the parity of $m_j(w)$ matters. Indeed $\Gamma_{j}^{m_{j}(w)}=I$ if $m_{j}(w)$ is even and $\Gamma_{j}^{m_{j}(w)}=\Gamma_j$ if $m_{j}(w)$ is odd. Hence, we define the parity set of the word $w$ by
\begin{align}
\pi(w) := \{j\in[2n]: m_j(w)\text{ is odd}\}, \label{eq:parity-set-def}
\end{align}
and, since the $\Gamma_i$ commute, we can express
\begin{align}
\Gamma_{i_1}\cdots\Gamma_{i_{2r}}=
\prod_{j=1}^{2n}\Gamma_j^{m_j(w)} =
\prod_{j\in\pi(w)}\Gamma_j =:
\Gamma_{\pi(w)}.
\end{align}
Let us note that $|\pi(w)|=2p$ for some $p\in\{0,\ldots, r\}$. Indeed, $\sum_{j}m_{j}(w)=2r$ by construction, therefore also the number of indices with odd multiplicity must be even.

For a fixed subset $S\subseteq[2n]$ with $|S|=2p$, define $N_{2r,2p}$ to be the number of words $w\in[2n]^{2r}$ whose parity set is exactly $S$. Thus,
\begin{align}
N_{2r,2p}
&:=
\#\{w\in[2n]^{2r}:\pi(w)=S\}. \label{eq:N-def}
\end{align}
By symmetry of the labels, this number depends only on $2r$ and $2p$, not on the particular subset $S$. Grouping the expansion of $\Lambda^{2r}$ according to the parity set gives
\begin{align}
\Lambda^{2r}
&=
\sum_{p=0}^r
N_{2r,2p}
\sum_{\substack{S\subseteq[2n]\\ |S|=2p}}
\Gamma_S. \label{eq:lambda-power-parity-grouping}
\end{align}
Taking expectation in $\rho^{\otimes2}$ we get
\begin{align}
    \tr(\Lambda^{2r}\rho^{\otimes 2})=\sum_{p=0}^{r}N_{2r,2p} \sum_{\substack{S\subseteq[2n]\\ |S|=2p}}
\tr(\Gamma_S\rho^{\otimes 2})= \sum_{p=0}^{r}N_{2r,2p} (-1)^{p(2p-1)}\sum_{\substack{S\subseteq[2n]\\ |S|=2p}}
a_{S}^2=\sum_{p=0}^{r}N_{2r,2p} (-1)^{p(2p-1)}A_{2p}
\end{align}
where we used \cref{eq:Asdef}. For an upper bound, we discard the signs:
\begin{align}
\operatorname{tr}\!\left(\Lambda^{2r}\rho^{\otimes2}\right)
&\leq
\sum_{p=0}^r
N_{2r,2p}A_{2p}. \label{eq:moment-discard-signs}
\end{align}
We bound separately $N_{2r,2p}$ and $A_{2p}$ in the \cref{lem1,lem2} and get:
\begin{align}
    \tr(\Lambda^{2r}\rho^{\otimes 2})\le \sum_{p=0}^{r}N_{2r,2p}A_{2p}\le \frac{(2r)!}{p!(r-p)!}n^{r-p}(en)^{p}=(2r)!n^{r}\sum_{p=0}^{r}\frac{1}{p!(r-p)!}e^{p}=\frac{(2r)!}{r!}n^{r}\sum_{p=0}^{r}\binom{r}{p}e^{p}=\frac{(2r)!}{r!}[(1+e)n]^{r}
\end{align}
which concludes the proof. \end{proof}

\begin{lemma}\label{lem1}
    For a fixed subset $S\subseteq[2n]$ with $|S|=2p$, define $N_{2r,2p}$ to be the number of words $w\in[2n]^{2r}$ whose parity set is exactly $S$. The following bound holds:
    \begin{align}
        N_{2r,2p}\le \frac{(2r)!}{(r-p)!}n^{r-p}
    \end{align}
    \begin{proof}
       We have the explicit expression, from which the bound follows immediately:
    \begin{align}
        N_{2r,2p}&=\sum_{\substack{\vec m\in \mathbb N^{2n}\\\sum_i m_i=2r\\\{i|m_i \in \text{odd}\}=S\\|S|=2p}}\frac{(2r)!}{m_1!\cdots m_{2n}!}\\
        &=\sum_{\substack{\vec a\in \mathbb N^{2n}\\\sum_i a_i=r-p\\}}\frac{(2r)!}{((2a_1+1)!\cdots (2a_{2p}+1)!(2a_{2p+1})!\cdots (2a_{2n})!}\\
        &\leq\sum_{\substack{\vec a\in \mathbb N^{2n}\\\sum_i a_i=r-p\\}}\frac{(2r)!}{2^{r-k}a_1!\cdots a_{2n}!}\\
        &\leq\frac{(2r)!}{2^{r-p}}\times \frac{(2n)^{r-p}}{(r-p)!}\\
        &=\frac{(2r)!}{(r-p)!}n^{r-p}\,.
    \end{align}
    \end{proof}
\end{lemma}

\begin{lemma}\label{lem2}
    For a subset $S=\{i_1<\ldots<i_s\}\subseteq [2n]$ and $\gamma_S=\gamma_{i_1}\cdot \gamma_{i_s}$. Let $P_S=i^{s(s-1)/2}\gamma_S$ its Hermitian version and $a_S=\tr(P_S\rho)$. Define $A_s = \sum_{\substack{S\subseteq[2n]\\ |S|=s}} a_S^2$. Then, for every integer $p$ it holds that
    \begin{align}
        A_{2p}\le \frac{(en)^p}{p!}
    \end{align}
    
    \begin{proof}
    To prove the lemma, we first show that for every $z\ge 0$, it holds
    \begin{align}
        \sum_{p=0}^{n}A_{2p}z^p\le (1+z)^{n}
    \end{align}
    We define the operator 
    \begin{align}
        B(z):=
\sum_{p=0}^n z^p
\sum_{\substack{S\subseteq[2n]\\ |S|=2p}}
P_S\otimes P_S.
    \end{align}
    and notice that by definition $\tr(B(z)\rho^{\otimes 2})=\sum_{p=0}^nA_{2p}z^p$. We now bound the operator norm of $B(z)$ as $\sum_{p=0}^nA_{2p}z^p\le \|B(z)\|$. Let us show that we can rewrite
    \begin{align}\label{rewriteB}
        B(z)=\frac{1}{2}\left[\prod_{i=1}^{2n}(I+i\sqrt{z}\Gamma_i)+\prod_{i=1}^{2n}(I-i\sqrt{z}\Gamma_i)\right]
    \end{align}
    Indeed, we have $ \prod_{i=1}^{2n}(I+i\sqrt{z}\Gamma_i)=\sum_{S\subseteq[2n]}(i\sqrt{z})^{|S|}\Gamma_S$ and therefore
    \begin{align}
\frac{1}{2}
\left[
\prod_{i=1}^{2n}(I+i\sqrt{z}\Gamma_i)
+
\prod_{i=1}^{2n}(I-i\sqrt{z}\Gamma_i)
\right]
&=
\sum_{\substack{S\subseteq[2n]\\ |S|\text{ even}}}
(i\sqrt{z})^{|S|}\Gamma_S=\sum_{p=0}^{n}z^p\sum_{\substack{S\subseteq[2n]\\|S|=2p}}(-1)^{p}\Gamma_S
\end{align}
The odd subsets vanish because, when $|S|$ is odd, $(it)^{|S|}+(-it)^{|S|}=0$. From \cref{gammaspsrelation} and putting $s=2p$ we have $\Gamma_{S}=(-1)^{p}P_{S}\otimes P_S$ because $    (-1)^{p(2p-1)}=(-1)^p$. This proves \cref{rewriteB}. It is now easy to bound the operator norm of $B(z)$. Indeed, all the $\Gamma_i$s commute and therefore they are simultaneously diagonalizable and have eigenvalues $\varepsilon_i=\pm 1$. Each eigenvalue of $B(z)$ reads
\begin{align}
    \left|\frac12
\left[
\prod_{i=1}^{2n}(1+i\sqrt z\,\epsilon_i)
+
\prod_{i=1}^{2n}(1-i\sqrt z\,\epsilon_i)
\right]\right|\le \prod_{i=1}^{2n}|1+i\sqrt{z}|=(1+z)^n
\end{align}
Hence $\|B(z)\|\le (1+z)^n$. 

Now we are ready to bound $A_{2p}$. For every $z\ge 0$, we know that $\sum_{p=0}^{n}A_{2p}z^p\le (1+z)^n$. Since $A_{2p}\ge 0$, for each $p\in[n]$, we have $A_{2p}z^{p}\le (1+z)^{n}$. Choosing $z=p/n$, we have 
        \begin{align}
            A_{2p}\left(\frac{p}{n}\right)^p\le (1+p/n)^n\le e^{p}\,.
        \end{align}
Using that $p!\le p^p$ and inverting the inequality shows the claim. 
    \end{proof}
\end{lemma}

\begin{corollary}[Tail bound] \label{cor:tailbound} Consider the operator $\Lambda^2/n$. For any state $\rho$ and $0<t<\frac{1}{16}$, we have
\begin{align}
    \tr(e^{t\frac{\Lambda^2}{n}}\rho^{\otimes 2})\le \frac{1}{1-16t}
\end{align}

\begin{proof}
    Thanks to \cref{lem:variancebounded}, we have that $n^{-r}\tr(\Lambda^{2r}\rho^{\otimes 2})\le \frac{(2r)!4^r}{r!}$. Hence, for any $0\le t\le 1/16$, denoting $X=\frac{\Lambda^2}{n}$ we have
    \begin{align}
        \tr(e^{tX}\rho^{\otimes 2})=\sum_{k=0}^{\infty}\frac{t^k}{k!}\tr(X^k\rho^{\otimes 2})\le \sum_{k=0}^{\infty}4^kt^k\frac{(2k)!}{(k!)^2}\le\sum_{k=0}^{\infty}(16t)^k=\frac{1}{1-16t}
    \end{align}
    where we used that $\frac{(2k)!}{(k!)^2}\le 4^k$.
\end{proof}
    
\end{corollary}

\section{Fannes-like inequality and asymptotic continuity of the fermionic entropy: proof of \cref{cor:robustness}}\label{cor:robustnessapp}
In this section, we prove \cref{cor:robustness}, i.e. the following continuity bound for the fermionic entropy:
\begin{align}
    |M_f(\rho)-M_f(\rho')|\le 2n\|\rho-\rho'\|_1
\end{align}
\begin{proof}
    First notice that $|M_f(\rho)-M_f(\rho')|\le n|P_f(\rho)-P_f(\rho')|$. Now define the state dependent observable $K_{\rho,\rho'}\coloneqq\frac{-i}{n}\sum_{i<j}(\Gamma_{ij}(\rho)+\Gamma_{ij}(\rho'))\gamma_i\gamma_j$ and notice that
    \begin{align}
        \tr(K_{\rho,\rho'}(\rho-\rho'))&=\frac{1}{n}\sum_{i<j}(\Gamma_{ij}(\rho)+\Gamma_{ij}(\rho'))(\Gamma_{ij}(\rho)-\Gamma_{ij}(\rho'))\\&=\frac{1}{n}\sum_{i<j}\Gamma_{ij}^2(\rho)-\Gamma_{ij}^2(\rho')\\&=\frac{1}{2n}(\|\Gamma(\rho)\|_2^2-\|\Gamma(\rho')\|_2^2)\\&=P_f(\rho)-P_f(\rho')
    \end{align}
    Hence, we can bound $|P_f(\rho)-P_f(\rho')|\le \|K_{\rho,\rho'}\|_{\infty}\|\rho-\rho'\|_1$. To bound the infinity norm of $K_{\rho,\rho'}$, notice that by linearity we can write
    \begin{align}
        K_{\rho,\rho'}=\frac{-2i}{n}\sum_{i<j}\Gamma_{ij}\left(\frac{\rho+\rho'}{2}\right)\gamma_i\gamma_j
    \end{align}
    Since $\Gamma\left(\frac{\rho+\rho'}{2}\right)$ is a valid correlation matrix, we know that there exists a Gaussian unitary $O$ such that $O\Gamma\left(\frac{\rho+\rho'}{2}\right) O^{T}=\bigoplus_{i=1}^{m}\begin{pmatrix}
        0 &c_i\\-c_i&0
    \end{pmatrix}$, where $c_i$ are its eigenvalues~\cite{BittelMeleEisertLeone2025}. Hence, we have:
    \begin{align}
        U_{O}K_{\rho,\rho'}U_{O}^{\dagger}=\frac{2}{n}\sum_{i=1}^{n}c_{i}Z_{i}
    \end{align}
    Since the operator norm is unitarily equivalent:
    \begin{align}
        \|K_{\rho,\rho'}\|_{\infty}=\frac{2}{n}\sum_{i=1}^n|c_{i}|\le 2
    \end{align}
    where we used the fact that $|c_i|\le 1$. The theorem follows. 
\end{proof}
\section{Enhanced tolerant property testing of fermionic Gaussian states, proof of \cref{cor:propertytesting}}\label{cor:propertytestingapp}

In this section, we present a algorithm that given $O(n\varepsilon_B^{-2}\log n\varepsilon_B^{-2})$ copies of an unknown and possibly mixed quantum state decides whether
\begin{itemize}
    \item {\bf Case A:} $\min_{\psi\in G}\|\rho-\sigma\|_1\le \varepsilon_A$;
    \item {\bf Case B:} $\min_{\psi\in G}\|\rho-\sigma\|_1> \varepsilon_B$;
\end{itemize}
where $G$ is the set of pure free fermionic quantum states, provided that $\varepsilon_A\le \frac{1}{16}\frac{\varepsilon_B^2}{4n}$. In Ref.~\cite{BittelMeleEisertLeone2025}, we have shown that
\begin{align}
    \min_{\psi}\|\rho-\psi\|_1\le \sqrt{2\sum_{i=1}^n(1-\lambda_i)}
\end{align}
where $\lambda_i$ are the eigenvalues of the correlation matrix $\Gamma(\rho)$. Noticing that $\|\Gamma(\rho)\|_2^2=2\sum_i\lambda_i^2$, it follows that 
\begin{align}
    \min_{\psi}\|\rho-\psi\|_1\le \sqrt{2n[1-P_f(\rho)]}
\end{align}

The testing algorithm proceeds as follows: we measure the operator $\frac{\Lambda^2}{n}$ using Bell sampling~\cite{Sierant_2026}. We recall that $\frac{1}{n}\tr(\Lambda^2\rho^{\otimes 2})=1-P_f(\rho)$. Since, $\tr[\Lambda^2\psi^{\otimes 2}]=0$ for any free-fermionic pure state, and in virtue of \cref{cor:robustness}, we have that if $\rho$ comes from case $A$, then $\frac{1}{n}\tr(\Lambda^2\rho^{\otimes 2})=1-P_f(\rho)=P_f(\psi)-P_f(\rho)\le 2\min_{\psi\in\mathcal{G}}\|\rho^{\otimes 2}-\psi^{\otimes 2}\|_1\le 4\varepsilon_A$. On the other hand, if $\rho$ comes from case $B$, we know that 
\begin{align}
    \varepsilon_B< \min_{\psi}\|\rho-\psi\|_1\le \sqrt{2n[1-P_f(\rho)]}\implies [1-P_f(\rho)]\ge \frac{\varepsilon_B^2}{2n}
\end{align}
Hence, to distinguish case A and case B it is sufficient to distinguish whether $\frac{1}{n}\tr(\Lambda^2\rho^{\otimes 2})\le \tilde{\varepsilon}_A$ or $\frac{1}{n}\tr(\Lambda^2\rho^{\otimes 2})> \tilde{\varepsilon}_B$, where $\tilde{\varepsilon}_A=4\varepsilon_A$ and $\tilde{\varepsilon}_B=\frac{\varepsilon_B^2}{2n}$. In order to achieve this task with $O(\tilde{\varepsilon}_B^{-1})$ samples, we need to do a few steps. First of all, denote $X=\frac{\Lambda^2}{2n}$ and define the new observable $Y\coloneqq\min\{X,L\mathbb{1}\}$ with $L$ a constant to be chosen later. Denote $\mathbb{E}[\cdot]\coloneqq\tr(\cdot\,\rho^{\otimes 2})$ for simplicity. Notice that
\begin{align}
    \mathbb{E}[X-Y]=\mathbb{E}[(X-L\mathbb{1})_{+}]
\end{align}
where $+$ denotes the positive part only. Then, by the tail-integral formula we can write the expectation value as
\begin{align}
    \mathbb{E}[(X-L\mathbb{1})_{+}]=\int_{L}^{x_{\max}}\Pr(X\ge x)dx
\end{align}
We can use Markov's inequality and \cref{cor:tailbound} to write:
\begin{align}
    \Pr(X\ge x)=\Pr(e^{tX}\ge e^{tx})=\Pr(e^{\frac{1}{32}X}\ge e^{\frac{1}{32}x})\le e^{-\frac{1}{32}x}\mathbb{E}[e^{\frac{1}{32}X}]\le 2e^{-\frac{1}{32}x}
\end{align}
Hence
\begin{align}
     \mathbb{E}[(X-L\mathbb{1})_{+}]\le 2\int_L^{\infty}e^{-\frac{1}{32}x}dx=64 e^{-L/32}
\end{align}
Let us choose $L$ such that
\begin{align}
    64 e^{-L/32}=\frac{\tilde{\varepsilon}_B}{2}\implies L=32\log\frac{128}{\tilde{\varepsilon}_B}\,.
\end{align}
Hence, our algorithm effectively measures the observable $Y$, corresponding to a measurement of the operator $\frac{\Lambda}{n}$, and rejects all outcomes larger than $L$ by replacing them with $L$. For case A, we have $\mathbb{E}[Y]\le \tilde{\varepsilon}_A$, and for case B we have $\mathbb{E}[Y]\ge \mathbb{E}[X]-64 e^{-L/32}=\frac{\tilde{\varepsilon}_B}{2}$. Imposing a promise gap $\tilde{\varepsilon}_A\le \frac{1}{4}\frac{\tilde{\varepsilon}_B}{2}$, it follows by standard arguments that measuring $Y$ yields a tester with sample complexity $O(\tilde{\varepsilon}_B^{-1}\log\tilde{\varepsilon}_B^{-1}\log\delta^{-1})$ and success probability at least $1-\delta$. This proves the theorem with a promise gap of $\varepsilon_A\le \frac{1}{16}\frac{\varepsilon_B}{4n}$ and a sample complexity of $O(n\varepsilon_B^{-2}\log(n\varepsilon_B^{-2})\log\delta^{-1})$.

\section{Unitary designs with doped Matchgate circuits: proof of \cref{cor:designs}}\label{cor:designapp}
In order to proof \cref{cor:designs}, we first need a preliminary lemma which constitutes an alternative way to measure the fermionic purity. 

\begin{lemma}
    Let $\mathcal{W}_{\beta}\coloneqq\exp(i\sqrt{\beta/(4n)}\Lambda)$. Then, for any $\beta\le 1$ it holds that 
    \begin{align}
\left|(1-\operatorname{Re}[\tr(\mathcal{W}_\beta\rho^{\otimes 2})])-\frac{\beta}{2}\left(1-P_f(\rho)\right)\right|\leq\beta^2,.
\end{align}
\end{lemma}
\begin{proof}
    Consider the unitary matrix $e^{i\Lambda\alpha}$ for $\alpha$ to be chosen later. Consider the expectation value $\tr(e^{i\Lambda\alpha}\rho^{\otimes 2})$, which can be estimated via a swap test. We can express
\begin{align}
    e^{-i\Lambda\alpha}=I+i\sum_{k=0}^{\infty}(-1)^{k}\frac{(\alpha\Lambda)^{2k+1}}{(2k+1)!}+\sum_{k=1}^{\infty}(-1)^{k}\frac{(\alpha\Lambda)^{2k}}{(2k)!}
\end{align}
Let us consider the real part of the expectation value:
\begin{align}
    \operatorname{Re}[\tr(e^{i\Lambda\alpha}\rho^{\otimes 2})]=1-\frac{1}{2}\alpha^2\tr(\Lambda^2\rho^{\otimes 2})+\sum_{k=2}^{\infty}(-1)^{k}\frac{(\alpha\Lambda)^{2k}}{(2k)!}
\end{align}
where, recalling that the correlation matrix is defined as $\Gamma_{ij}=-i\tr(\gamma_i\gamma_j\rho)$ for $i\neq j$ then:
\begin{align}
    \tr(\Lambda^2\rho^{\otimes 2})=\sum_{i,j}\tr^2(\gamma_{i}\gamma_{j}\rho)=2n-\sum_{i,j}\Gamma_{ij}^2=2n-\|\Gamma\|_{2}^{2}
\end{align}
Let us bound the rest:
\begin{align}
    \left|\sum_{k=2}^{\infty}(-1)^{k}\frac{\alpha^{2k}\tr(\Lambda^{2k}\rho^{\otimes 2})}{(2k)!}\right|\le \sum_{k=2}^{\infty}\frac{\alpha^{2k}}{(2k)!}(4n)^{k}\frac{(2k)!}{k!}=\sum_{k=2}^{\infty}\frac{(4\alpha^2n)^k}{k!}=e^{4\alpha^2n}-1-4\alpha^2n
\end{align}
Imposing $4\alpha^2n\le 1$, we have $e^{4\alpha^2n}-1-4\alpha^2n\le (4\alpha^2n)^2$ by Heinz inequality. Hence, we impose $\alpha^2= \frac{\beta}{4n}$ to get
\begin{align}
    \left|\operatorname{Re}[\tr(e^{i\Lambda\sqrt{\beta/(4n)}}\rho^{\otimes 2})]-\left(1-\frac{\beta}{4}\left(1-\frac{\|\Gamma\|_2^2}{2n}\right)\right)\right|\le  \beta^2
\end{align}
valid for any $\beta=4\alpha^2n\le 1$. This concludes the proof.
\end{proof}

{\em Proof of \cref{cor:designs}.} First, by definition of state $2$-design in \cref{statedesigns} we have that
        \begin{align}
            \|\Phi_{\mathcal{E}_t}(\ketbra{0}{0}^{\otimes 2})-\Phi_{\haar}(\ketbra{0}{0}^{\otimes 2})\|_1&\ge\tr[\mathrm{Re}\mathcal{W}_{\beta}\Phi_{\mathcal{E}_t}(\ketbra{0}{0}^{\otimes 2})]-\tr[\mathrm{Re}\mathcal{W}_{\beta}\Phi_{\haar}(\ketbra{0}{0}^{\otimes 2})]\\&\ge \frac{\beta}{4}\left[\mathbb{E}_{\psi\sim \mathcal{E}_t}P_f(\psi)-\mathbb{E}_{\psi\sim \haar}P_f(\psi)\right]-2\beta^2
        \end{align}
        Let us evaluate the two terms separately. First, $\Phi_{\haar}(\ketbra{0}{0}^{\otimes 2})=\frac{1+T}{d(d+1)}$ where $T$ is the swap operator defined on two copies of the Hilbert space. Hence:
        \begin{align}
            \mathbb{E}_{\psi\sim \haar}P_f(\psi)=\frac{1}{d(d+1)2n}\sum_{i\neq j}(\tr(\gamma_i\gamma_j)^2+\tr(\gamma_i\gamma_j\gamma_i\gamma_j))=\frac{(2n-1)}{d+1}
        \end{align}
        On the other hand, let $\ket{\psi_t}=U_t\ket{0}$ with $U_t\in\mathcal{E}_t$. It follows that~\cite{MeleHerasymenko2025} there exists a Matchgate unitary $U\in\mathcal{M}_n$ such that $U\ket{\psi_t}=\ket{0}^{\otimes n-\kappa t}\otimes \ket{\phi}_{\kappa t}$. Using the additivity of the Frobenius norm of the correlation matrix, it follows that
        \begin{align}
            \mathbb{E}_{\psi\sim \mathcal{E}_t}P_f(\psi)\ge \frac{(n-\kappa t)}{n}
        \end{align}
        Therefore
        \begin{align}
            \|\Phi_{\mathcal{E}_t}(\ketbra{0}{0}^{\otimes 2})-\Phi_{\haar}(\ketbra{0}{0}^{\otimes 2})\|_1\ge \frac{\beta}{4}\left(1-\frac{\kappa t}{n}-\frac{2n-1}{d+1}\right)-2\beta^2
        \end{align}
        Since this holds for any $\beta\le 1$, optimizing over $\beta$ gives
        \begin{align}
            \|\Phi_{\mathcal{E}_t}(\ketbra{0}{0}^{\otimes 2})-\Phi_{\haar}(\ketbra{0}{0}^{\otimes 2})\|_1\ge \frac{1}{128}\left(1-\frac{\kappa t}{n}-\frac{2n-1}{d+1}\right)^2
        \end{align}
        Let us impose for simplicity $\frac{2n-1}{d+1}\le \frac{1}{n}$ which holds for any $n\ge 7$ so that 
        \begin{align}
            \|\Phi_{\mathcal{E}_t}(\ketbra{0}{0}^{\otimes 2})-\Phi_{\haar}(\ketbra{0}{0}^{\otimes 2})\|_1\ge \frac{1}{128}\left(1-\frac{4\kappa t}{n}\right)
        \end{align}
        which implies that whenever $\kappa t\le n/8$ then $\|\Phi_{\mathcal{E}_t}(\ketbra{0}{0}^{\otimes 2})-\Phi_{\haar}(\ketbra{0}{0}^{\otimes 2})\|_1\ge 1/256$ proving the lower bound. \qed

\end{document}